\documentclass[11pt]{paper}

\usepackage{fullpage}
\usepackage{amssymb}
\usepackage{amsmath}
\usepackage{amsthm}
\usepackage[usenames]{color}

\usepackage{enumerate}
\usepackage{cite}
\usepackage{float}
\usepackage{xspace}

\floatstyle{ruled}
\newfloat{algorithm}{thp}{alg}
\floatname{algorithm}{Algorithm}

\newcommand{\mathset}[1]{\ensuremath {\mathbb {#1}}}
\newcommand{\N}{\mathset {N}}
\newcommand{\R}{\mathset {R}}

\newcommand{\EX}{\mathbf {E}}

\newcommand{\eqdef}{:=}

\newcommand{\DKL}{D_\textup{KL}}

\newtheorem{theorem}{Theorem}[section]

\newtheorem{lemma}[theorem]{Lemma}

\newtheorem{corollary}[theorem]{Corollary}
\newtheorem{claim}[theorem]{Claim}

\title{Five Proofs of Chernoff's Bound with 
Applications\footnote{Supported in part by DFG Grants MU 3501/1 and
MU 3501/2 and ERC StG 757609.}}
\author{Wolfgang Mulzer}
\institution{Institut f\"ur Informatik\\ Freie Universit\"at Berlin\\
E-Mail: \texttt{mulzer@inf.fu-berlin.de}}

\begin{document}

\maketitle

\begin{abstract}
We discuss five ways of proving Chernoff's bound and show
how they lead to different extensions of the basic bound.
\end{abstract}

\section{Introduction}

Chernoff's bound gives an estimate on the probability
that a sum of independent Binomial random variables
deviates from its expectation~\cite{Hoeffding63}.
It has many variants and extensions that are known under various 
names such as
\emph{Bernstein's inequality} or \emph{Hoeffding's 
bound}~\cite{Bernstein64,Hoeffding63}.
Chernoff's bound is one of the most basic and versatile
tools in the life of a theoretical computer 
scientist,
with a seemingly endless amount of applications.
Almost every contemporary textbook on algorithms or 
complexity theory contains a statement and a proof of the 
bound~\cite{AroraBa09,Goldreich08,KleinbergTa06,CormenLeRiSt09}, 
and there are several texts that discuss its various applications 
in great detail (e.g., the textbooks by 
Alon and Spencer~\cite{AlonSp16}, 
Dubhashi and Panchonesi~\cite{DubhashiPa09},
Mitzenmacher and Upfal~\cite{MitzenmacherUp17},
Motwani and Raghavan~\cite{MotwaniRa95}, 
or the articles by 
Chung and Lu~\cite{ChungLu06},
Hagerup and R\"ub~\cite{HagerupRu90}, or
McDiarmid~\cite{McDiarmid98}).

In the present survey, we will see five different
ways of proving the basic Chernoff bound. The different techniques
used in these proofs allow various generalizations and extensions,
some of which we will also discuss.

\section{The Basic Bound}

We begin with a statement of the basic Chernoff bound.
For this, we first need a notion from information theory~\cite{CoverTo06}.
Let $P = (p_1, \dots, p_m)$ and $Q = (q_1, \dots, q_m)$ be two 
probability distributions 
on $m$ elements, i.e., $p_i, q_i \in \R$ with 
$p_i, q_i \geq 0$, for $i=1, \dots, m$,  
and $\sum_{i=1}^{m} p_i = \sum_{i=1}^m q_i = 1$. The
\emph{Kullback-Leibler divergence} or \emph{relative entropy} of 
$P$ and $Q$ is defined as
\[
  \DKL(P \| Q) \eqdef \sum_{i=1}^m p_i \ln \frac{p_i}{q_i}.
\]
If $m = 2$, i.e., if $P = (p, 1-p)$ and $Q = (q, 1-q)$, we write
$\DKL(p \| q)$ for $\DKL((p, 1-p) \| (q, 1-q))$.
The Kullback-Leibler divergence measures the distance between the
distributions $P$ and $Q$: it represents the expected loss of efficiency 
if we encode an $m$-letter alphabet with distribution $P$ with a code 
that is optimal
for distribution $Q$. Now, the basic Chernoff bound is as follows:
\begin{theorem}\label{thm:chernoff}
Let $n \in \N$, $p \in [0,1]$, and let $X_1, \dots, X_n$ be $n$ independent 
random variables with $X_i \in \{0,1\}$
and $\Pr[X_i = 1] = p$, for $i = 1, \dots n$. 
Set $X \eqdef \sum_{i=1}^n X_i$.
  Then, for any $t \in [0, 1-p]$, we have
  \[
    \Pr[X \geq (p+t)n ] \leq e^{-\DKL(p+t \| p)n}.
  \]
\end{theorem}

\section{Five Proofs for Theorem~\ref{thm:chernoff}}\label{sec:proofs}

We will now see five different ways of proving Theorem~\ref{thm:chernoff}.

\subsection{The Moment Method}\label{sec:moment}

The usual textbook proof of Theorem~\ref{thm:chernoff} uses the 
exponential function $\exp$ and Markov's inequality. 
It is called the \emph{moment method},
because $\exp$ simultaneously encodes all \emph{moments}
$X, X^2, X^3, \dots$ of $X$.
This trick is often attributed to Bernstein~\cite{Bernstein64}.
It is very general and 
can be used to obtain several variants of Theorem~\ref{thm:chernoff},
perhaps most prominently, the Azuma-Hoeffding inequality for
martingales with bounded differences~\cite{Hoeffding63,Azuma67}.

The proof goes as follows.
Let $\lambda > 0$ be a parameter to be determined later. We have
\[
  \Pr[X \geq (p+t)n ] = \Pr[\lambda X \geq \lambda (p+t)n ] = 
  \Pr\bigl[e^{\lambda X} \geq e^{\lambda (p+t)n} \bigr].
\]
From Markov's inequality, we obtain
\[
  \Pr\bigl[e^{\lambda X} \geq e^{\lambda (p+t)n} \bigr] \leq
    \frac{\EX[e^{\lambda X}]}{e^{\lambda (p+t)n}}.
\]
Now, the independence of the $X_i$ yields
\[
  \EX[e^{\lambda X}] = \EX\Bigl[e^{\lambda \sum_{i=1}^n X_i}\Bigr]
  =  \EX\Biggl[\prod_{i=1}^n e^{\lambda X_i}\Biggr]
  =  \prod_{i=1}^n \EX\Bigl[e^{\lambda X_i}\Bigr]
  =  \bigl(p e^\lambda + 1-p\bigr)^n.
\]
Thus, 
\begin{equation}\label{equ:lambdabound}
  \Pr[X >  (p+t)n] \leq
    \Bigl(\frac{pe^\lambda + 1-p}{e^{\lambda (p+t)}}\Bigr)^n,
\end{equation}
for every $\lambda > 0$. Optimizing for $\lambda$ using calculus,
we get that the right hand side is minimized if
\[
 e^\lambda = \frac{(1-p)(p+t)}{p(1-p-t)}.
\]
Plugging this into (\ref{equ:lambdabound}), we get
\[
  \Pr[X >  (p+t)n] \leq \Biggl[\Bigl(\frac{p}{p+t}\Bigr)^{p+t}
         \Bigl(\frac{1-p}{1-p-t}\Bigr)^{1-p-t}\Biggr]^n = 
      e^{-\DKL(p+t \| p)n},
\]
as desired.

\subsection{Chv\'atal's Method}\label{sec:chvatal}

The following proof of Theorem~\ref{thm:chernoff} is due to
Chv\'atal~\cite{Chvatal79}. As we will see below,
it can be generalized to give tail bounds for the 
\emph{hypergeometric distribution}.
Let $B(n,p)$ be the random variable that gives the number of heads in 
$n$ independent Bernoulli trials with success probability $p$.
Then,
\[
\Pr[B(n,p) = l] = \binom{n}{l} p^l (1-p)^{n-l},
\]
for $l = 0, \dots, n$. Thus, for any $\tau \geq 1$ and $k \geq pn$, we get
\begin{multline*}
\Pr[B(n,p)\ge k] = \sum_{i=k}^n \binom{n}{i}p^i (1-p)^{n-i}\\ 
\leq \sum_{i=k}^n \binom{n}{i}p^i (1-p)^{n-i} 
\underbrace{\tau^{i-k}}_{\geq 1} +
\underbrace{\sum_{i=0}^{k-1} \binom{n}{i}p^i (1-p)^{n-i} 
\tau^{i-k}}_{\geq 0} 
= \sum_{i=0}^n \binom{n}{i}p^i (1-p)^{n-i} \tau^{i-k}.
\end{multline*}
Using the Binomial theorem, we obtain
\[
\Pr[B(n,p)\ge k] \leq
\sum_{i=0}^n \binom{n}{i}p^i (1-p)^{n-i} \tau^{i-k} = 
\tau^{-k}\sum_{i=0}^n \binom{n}{i}(p\tau)^{i} (1-p)^{n-i} =
\frac{(p\tau+1-p)^n}{\tau^k}.
\]
If we write $k = (p+t)n$ and $\tau = e^\lambda$, we get
\[
\Pr[B(n,p)\ge (p+t)n] \leq
\Bigl(\frac{p e^\lambda+1-p}{e^{\lambda(p+t)}}\Bigr)^n.
\]
This is the same as (\ref{equ:lambdabound}), so we can complete the
proof of Theorem~\ref{thm:chernoff} as in Section~\ref{sec:moment}.

\subsection{The Impagliazzo-Kabanets Method}\label{sec:ik}

The third proof is due to Impagliazzo and Kabanets~\cite{ImpagliazzoKa10},
and it leads to a constructive version of the bound.
Let $\lambda \in [0,1]$ be a parameter to be chosen later. 
Let $I \subseteq \{1, \dots, n\}$ be a random index set obtained by
including each element $i \in \{1, \dots, n\}$ with probability 
$\lambda$.
We estimate 
$\EX\bigl[\prod_{i \in I} X_i\bigr]$ in two different ways, where
the expectation is over the random choice of $X_1, \dots, X_n$ and $I$.

On the one hand, using the law of total expectation and independence, we have
\begin{multline}\label{equ:ikupper}
\EX\Bigl[\prod_{i \in I} X_i\Bigr]
= \sum_{S \subseteq \{1, \dots, n\}} \Pr[I = S] \cdot  
\EX\Bigl[\prod_{i \in S} X_i\Bigr]
= \sum_{S \subseteq \{1, \dots, n\}} \Pr[I = S] \cdot  
    \prod_{i \in S} \Pr[ X_i = 1]\\
= \sum_{S \subseteq \{1, \dots, n\}} \lambda^{|S|}(1-\lambda)^{n-|S|} \cdot 
    p^{|S|}
= (\lambda p + 1 - \lambda)^n.
\end{multline}
On the other hand, 
by the law of total expectation,
\[
\EX\Bigl[\prod_{i \in I} X_i\Bigr]
 \geq
\EX\Bigl[\prod_{i \in I} X_i  \mid X \geq (p+t)n\Bigr]\Pr[X \geq (p+t)n].
\]
Now, fix $X_1, \dots, X_n$ with $X \geq (p+t)n$. For the fixed choice
of $X_1 = x_1, \dots, X_n = x_n$, the expectation 
$\EX\bigl[\prod_{i \in I} x_i\bigr]$ is exactly the probability that
$I$ avoids all the $n-X$ indices $i$ where $x_i = 0$.
Thus, the conditional expectation is
\[
\EX\Bigl[\prod_{i \in I} X_i  \mid X \geq (p+t)n\Bigr]
= 
\EX\Bigl[(1-\lambda)^{n-X} \mid X \geq (p+t)n\Bigr]
 \geq (1-\lambda)^{(1-p-t)n},
\]
so
\[
\EX\Bigl[\prod_{i \in I} X_i\Bigr]
 \geq
(1-\lambda)^{(1-p-t)n}\Pr[X \geq (p+t)n].
\]
Combining with (\ref{equ:ikupper}),
\begin{equation}\label{equ:ikbound}
\Pr[X \geq (p+t)n] \leq 
\left(\frac{\lambda p + 1 - \lambda}{(1-\lambda)^{(1-p-t)}}\right)^n.
\end{equation}
Using calculus, we get that the right hand side is minimized for
$\lambda = t/(1-p)(p+t)$ (note that $\lambda \leq 1$ for $t \leq 1-p$).
Plugging this into (\ref{equ:ikbound}),
\[
  \Pr[X >  (p+t)n] \leq
    \Biggl[\Bigl(\frac{p}{p+t}\Bigr)^{p+t}
      \Bigl(\frac{1-p}{1-p-t}\Bigr)^{1-p-t}\Biggr]^n = 
      e^{-\DKL(p+t \| p)n},
\]
as desired.

\subsection{The Encoding Argument}\label{sec:encoding}

The next proof stems from discussions with 
Luc Devroye, G\'abor Lugosi, and Pat Morin,
and it is inspired by an encoding argument~\cite{MorinMuRe17}.
A similar argument can also be derived from Xinjia Chen's
\emph{likelihood ratio method}~\cite{Chen13}.
Let $\{0,1\}^n$ be the set of all bit strings of length $n$,
and let $w: \{0, 1\}^n \rightarrow [0,1]$ be a
\emph{weight function}. We call $w$
\emph{valid} if $\sum_{x \in \{0,1\}^n} w(x) \leq 1$.
The following lemma says that for any probability distribution
$p_x$ on $\{0,1\}^n$, a valid weight function 
is unlikely to be substantially larger than $p_x$.
\begin{lemma}\label{lem:encoding}
Let $\mathcal{D}$ be a probability 
distribution on $\{0,1\}^n$ that assigns to each $x \in \{0,1\}^n$
a probability $p_x$, and let $w$ be a valid 
weight function. 
For any $s \geq 1$, we have 
\[
  \Pr_{x \sim \mathcal{D}}\left[w(x) \geq sp_x\right] 
     \leq 1/s.
\]
\end{lemma}

\begin{proof}
Let $Z_{s} = \{ x \in \{0,1\}^n \mid w(x) \geq sp_x\}$.
We have
\[
  \Pr_{x \sim \mathcal{D}}\left[w(x) \geq s p_x\right] 
  = \sum_{\substack{x \in Z_s \\ p_x > 0}}  p_x
  \leq \sum_{\substack{x \in Z_s \\ p_x > 0}} p_x \frac{w(x)}{sp_x}
  \leq  (1/s) \sum_{x \in Z_s} w(x) 
    \leq 1/s,
\]
since $w(x) / sp_x  \geq 1$ for $x \in Z_s$, $p_x > 0$,  and since
$w$ is valid.
\end{proof}

We now show that Lemma~\ref{lem:encoding} 
implies Theorem~\ref{thm:chernoff}. 
For this, we interpret the sequence $X_1, \dots, X_n$
as a bit string of length $n$. This induces a probability distribution 
$\mathcal{D}$ that assigns to each $x \in \{0, 1\}^n$ the 
probability 
$p_x = p^{k_x} (1-p)^{n-k_x}$, where $k_x$ denotes the number of
$1$-bits in $x$.
We define a weight function $w : \{0,1\}^n \rightarrow [0,1]$ by
$w(x) = (p+t)^{k_x}(1-p-t)^{n-k_x}$, for 
$x \in \{0,1\}^n$.
Then $w$ is valid, since $w(x)$ is the
probability that $x$ is generated by setting each bit 
to $1$ independently with probability $p+t$.
For $x \in \{0,1\}^n$, 
we have
\[
\frac{w(x)}{p_x}
=
\left(\frac{p+t}{p}\right)^{k_x}
\left(\frac{1-p-t}{1-p}\right)^{n - k_x}. 
\]
Since $((p+t)/p)((1-p)/(1-p-t)) \geq 1$, it follows
that $w(x)/p_x$ is an increasing function of $k_x$.
Hence, if $k_x \geq (p + t)n$, we have
\[
\frac{w(x)}{p_x}
\geq
\left[\left(\frac{p+t}{p}\right)^{p+t}
\left(\frac{1-p-t}{1-p}\right)^{1 - p-t}\right]^n
= e^{\DKL(p+t \| p)n}.
\]
We now apply Lemma~\ref{lem:encoding} to $\mathcal{D}$ and $w$ to get
\[
\Pr[X \geq (p+t)n] = \Pr_{x\sim \mathcal{D}} [ k_x \geq (p+t)n]
\leq \Pr_{x \sim \mathcal{D}} \left[ w(x) \geq 
p_xe^{\DKL(p+t \| p)n}\right]
    \leq e^{-\DKL(p+t \| p)n},
\]
as claimed in Theorem~\ref{thm:chernoff}.

See the survey~\cite{MorinMuRe17} for a more thorough
discussion of how this proof is related to coding theory.

\subsection{A Proof via Differential Privacy}

The fifth proof of Chernoff's bound is due to
Steinke and Ullman~\cite{SteinkeUl17}, and it uses
methods from the theory of differential privacy~\cite{DworkRo14}.
Unlike the previous four proofs, it seems to lead to
a slightly weaker version of the bound.
Let $m$ be a parameter to be determined later.
The main idea is to bound the expectation of $m - 1$ independent
copies of $X$.

\begin{lemma}Let $m \in \N$ and $m \leq e^{n}$. 
Let $X^{(1)}, \dots, X^{(m - 1)}$ be $m - 1$ independent copies of 
$X$, and set $X^{(m)} = \EX[X]$.
Then, 
\[
  \EX\big[\max\{X^{(1)}, \dots, X^{(m)}\}\big] 
  \leq pn + 
  5\sqrt{n \ln m}.
\]
\label{lem:maxexp}
\end{lemma}

We will give a proof of Lemma~\ref{lem:maxexp} below.
First, however, we will see how we can use
Lemma~\ref{lem:maxexp} to derive the following
weaker version of Theorem~\ref{thm:chernoff}.\footnote{In the published 
version of this paper, the proof of 
Theorem~\ref{thm:weakchernoff} is based
on an incorrect application of Markov's inequality.
We have changed Lemma~\ref{lem:maxexp} so that
$X^{(m)}$ is fixed to $\EX[X]$.
This ensures that Markov's inequality is applied to a nonnegative
random variable. We thank Natalia Shenkman for pointing this out to us.}

\begin{theorem}\label{thm:weakchernoff}
Let $n \in \N$, $p \in [0,1]$, and let $X_1, \dots, X_n$ be $n$ independent 
random variables with $X_i \in \{0,1\}$
and $\Pr[X_i = 1] = p$, for $i = 1, \dots n$. 
Set $X \eqdef \sum_{i=1}^n X_i$.
Then, for any $t \in [0, 1-p]$, we have
  \[
 \Pr[X \geq (p+t)n ] \leq e^{1-\frac{1}{64}t^2 n}.
  \]
\end{theorem}

\begin{proof}%
We may assume that $t \geq 8/\sqrt{n}$, since otherwise the lemma
holds trivially.
Set $\alpha = \Pr[X \geq (p+t)n ]$.
Let $X^{(1)}, \dots, X^{(m - 1)}$ be $m - 1$ independent 
copies of $X$ and let $X^{(m)} = \EX[X]$. Then,
\begin{equation}\label{equ:maxXlb}
\Pr\big[\max\{X^{(1)}, \dots, X^{(m)}\} 
\geq (p+t)n\big] = 1 - (1-\alpha)^{m - 1} \geq 1 - e^{-\alpha (m - 1)}.
\end{equation}
On the other hand, Markov's inequality gives
\begin{multline*}
\Pr\big[\max\{X^{(1)}, \dots, X^{(m)}\} \geq (p+t)n\big] = 
\Pr\big[\max\{X^{(1)}, \dots, X^{(m)}\} -pn \geq tn\big]\\ 
\leq 
\frac{\EX\big[\max\{X^{(1)}, \dots, X^{(m)}\} - pn\big]}
{tn} 
\leq \frac{5\sqrt{\ln m}}{t\sqrt{n}},
\end{multline*}
by Lemma~\ref{lem:maxexp}.
Thus, setting $m = \exp\Big(\big(\frac{e-1}{5e}\big)^2t^2n\Big)$, and 
combining with
(\ref{equ:maxXlb}), we get
\[
\frac{e-1}{e} \geq 1 - e^{-\alpha (m - 1)} \Leftrightarrow
\alpha \leq \frac{1}{\exp\Big(\big(\frac{e-1}{5e}\big)^2t^2n\Big)
- 1}
\leq
\frac{1}{\exp\big(\frac{t^2n}{64}\big) - 1},
\]
since $\big(\frac{e-1}{5e}\big)^2 \geq \frac{1}{64}$.
Now the lemma follows from
\[
\frac{\exp\big(\frac{t^2n}{64}\big)}{\exp\big(\frac{t^2n}{64}\big) - 1}
\leq
\frac{e}{e-1}
\leq 
e,
\]
which holds as $t \geq 8/\sqrt{n}$, as $x \mapsto x/(x-1)$ is
decreasing for $x \geq 0$, and as $e \geq 2$.
\end{proof}

It remains to prove Lemma~\ref{lem:maxexp}.
For this, we use an idea from differential privacy.
Let $A \in [0, 1]^{m \times n}$, $A = (a_{ij})$, 
be an $(m \times n)$-matrix
with entries from $[0, 1]$. For a given parameter $\gamma > 1$,
we define a random variable 
$S_\gamma(A)$ with values in $\{1, \dots, m \}$ as follows: 
for $i = 1, \dots, m$,
let $b_i = \sum_{j = 1, \dots, n} a_{ij}$ be the sum of the
entries in the $i$-th row of $A$. Set 
\[
C_\gamma(A) = \sum_{i = 1}^{m} \gamma^{b_i}.
\]
Then, for $i = 1, \dots, m$, we define
\[
\Pr[S_\gamma(A) = i] = \frac{\gamma^{b_i}}{C_\gamma(A)}.
\]
The random variable $S_\gamma(A)$ is called a \emph{stable selector}
for $A$ (see the work by McSherry and Talwar~\cite{McSherryTa07} for
more background).
The next lemma states two interesting properties for
$S_\gamma(A)$.
For a matrix $A \in [0, 1]^{m \times n}$, 
a vector $\vec{c} \in [0,1]^m$,  and a number $j \in \{1, \dots, n\}$
we denote by $(A_{-j}, \vec{c})$ 
the matrix obtained from $A$ by replacing the $j$-th column of $A$ with
$\vec{c}$. 

\begin{lemma}\label{lem:stableselect}
Let $A \in [0, 1]^{m \times n}$ be an $m \times n$ matrix
with entries in $[0, 1]$. We have
\begin{itemize}
\item \textbf{Stability}:
For every vector $\vec{c} \in [0,1]^m$ and every
$i \in \{1, \dots, m\}$, 
\[
\gamma^{-2} \Pr[S_\gamma(A_{-j},\vec{c}) = i]
\leq \Pr[S_\gamma(A) = i] \leq 
\gamma^{2} \Pr[S_\gamma(A_{-j},\vec{c}) = i].
\]
\item \textbf{Accuracy}:
Let
$b_i$ be the sum of the $i$-th row of $A$. Then,

\[
\EX_{i \sim S_\gamma(A)}[b_i] \leq 
\max_{i = 1}^{m} b_i  \leq \EX_{i \sim S_\gamma(A)}[b_i] + \log_\gamma m. 
\]
\end{itemize}
\end{lemma}
\begin{proof}
\textbf{Stability}:
for $k \in \{1, \dots, m\}$, let $b_k$ be the sum of the
$k$-th row of $A$, and let $\widetilde{b}_k$ be the sum of the
$k$-th row of $(A_{-j}, \widetilde{c})$. Since $A$ and 
$(A_{-j}, \widetilde{c})$ differ in one column, and since the
entries are from $[0, 1]$, we have 
$\widetilde{b}_k - 1 \leq b_k \leq \widetilde{b}_k + 1$. Hence, 
\[
\gamma^{-1} C_\gamma(A_{-j},\vec{c}) \leq C_\gamma(A) \leq
\gamma C_\gamma(A_{-j},\vec{c})
\]
and
\[
\gamma^{-2} \Pr[S_\gamma(A_{-j},\vec{c}) = i]
\leq \Pr[S_\gamma(A) = i] \leq 
\gamma^{2} \Pr[S_\gamma(A_{-j},\vec{c}) = i],
\]
as claimed.

\textbf{Accuracy}:
The inequality
$\EX_{i \sim S_\gamma(A)}[b_i] \leq \max_{i = 1}^{m} b_i$
is obvious. For the second inequality, we observe that
by definition, 
\[
  b_i = \log_\gamma (C_\gamma(A) \Pr[S_\gamma(A) = i]).
\]
Thus,
\begin{align*}
\EX_{i \sim S_\gamma(A)}[b_i]
&= \sum_{i = 1}^{m}\Pr[S_\gamma(A) = i]\log_\gamma 
(C_\gamma(A) \Pr[S_\gamma(A) = i])\\
&=
\sum_{i = 1}^{m}\Pr[S_\gamma(A) = i]\log_\gamma 
C_\gamma(A) - 
\sum_{i = 1}^{m}\Pr[S_\gamma(A) = i]\log_\gamma 
\frac{1}{\Pr[S_\gamma(A) = i]}\\
&\geq
\sum_{i = 1}^{m}\Pr[S_\gamma(A) = i]\log_\gamma 
\gamma^{\max_{i = 1}^m b_i} - \log_\gamma m,\\
&=
\max_{i = 1}^m b_i- \log_\gamma m,
\end{align*}
since 
$C_\gamma(A) = \sum_{i = 1}^{m} \gamma^{b_i} \geq 
\gamma^{\max_{i = 1}^m b_i}$
and since $x \mapsto - \log_\gamma(x)$ is a convex function.
\end{proof}

Lemma~\ref{lem:stableselect} shows that $S_\gamma(A)$ 
constitutes a reasonable mechanism of estimating the maximum row sum
of $A$ without revealing too much information about any single
column of $A$. We can now use
Lemma~\ref{lem:stableselect} to bound the expectation of
the maximum of $m - 1$ independent copies of $X$ and 
$\EX[X]$.
\begin{lemma}\label{lem:maxgamma}
Let $m \in \N$. 
let $X^{(1)}, \dots, X^{(m - 1)}$ be $m - 1$ independent copies of 
$X$, and set $X^{(m)} = \EX[X]$.
Then, for any $\gamma > 1$, we have
\[
  \EX\big[\max\{X^{(1)}, \dots, X^{(m)}\}\big] \leq 
  \gamma^2pn + \log_\gamma m.
\]
\end{lemma}
\begin{proof}
Let $X_1^{(1)}, \dots, X_1^{(m - 1)}$ be $m - 1$ independent
copies of $X_1$, and let $X_1^{(m)} = \EX[X_1]$; let
$X_2^{(1)}, \dots, X_2^{(m - 1)}$ be $m - 1$ independent copies
of $X_2$ and let $X_2^{(m)} = \EX[X_2]$; and so on.
We consider the random $m \times n$ matrix $M \in \{0,1\}^{m \times n}$ 
whose entry in row $i$ and column $j$ is $X_j^{(i)}$.
Then, we can write $X^{(i)} = \sum_{j = 1}^n X_j^{(i)}$, for
$i = 1, \dots, m$. By the accuracy claim in Lemma~\ref{lem:stableselect},
\begin{equation}\label{equ:accumatrix}
\EX_{M} \big[\max\{X^{(1)}, \dots, X^{(m)}\}\big]
\leq \EX_{M, i \sim S_\gamma(M)}\big[X^{(i)}\big] + \log_\gamma m
\end{equation}
Now we bound $\EX_{M, i \sim S_\gamma(M)}\big[X^{(i)}\big]$. 
We unwrap the expectation for $i \sim S_\gamma(M)$
and get
\[
\EX_{M, i \sim S_\gamma(M)}[X^{(i)}]
=
\EX_{M}
\Big[\sum_{i = 1}^{m} \Pr[S_\gamma(M) = i] X^{(i)}\Big]
\]
Let $\widetilde{M}$ be an independent copy of $M$.
Denote the entry in the $i$-th row and $j$-th column
of $\widetilde{M}$ by $\widetilde{X}_j^{(i)}$, and
set $\widetilde{X}^{(i)} = \sum_{j = 1}^{n} \widetilde{X}_j^{(i)}$,
for $i = 1, \dots, m$. 
By the stability claim in Lemma~\ref{lem:stableselect}, for
every $j \in \{1, \dots, n\}$,
\begin{align*}
\EX_{M}
\Big[\sum_{i = 1}^{m} \Pr\big[S_\gamma(M) = i\big] X^{(i)}\Big]
&\leq
\gamma^2
\EX_{M,\widetilde{M}}
\Big[\sum_{i = 1}^{m} 
\Pr\big[S_\gamma(M_{-j},\widetilde{M}_j) = i\big] X^{(i)}\Big].\\
\intertext{Since the random variables $X_j^{(i)}$, $\widetilde{X}_j^{(i)}$,
$1 \leq i \leq m$, $1 \leq j \leq n$,
are independent, the pairs $\big((M_{-j}, \widetilde{M}_j), X_j^{(i)}\big)$
and $\big(M, \widetilde{X}_j^{(i)}\big)$ have the same distribution. 
Therefore, we can write}
\EX_{M}
\Big[\sum_{i = 1}^{m} 
\Pr\big[S_\gamma(M) = i\big] X^{(i)}\Big]&=
\EX_{M}
\Big[\sum_{i = 1}^{m} 
\sum_{j = 1}^{n}
\Pr\big[S_\gamma(M) = i\big] X_j^{(i)}\Big]\\
&\leq \gamma^2
\EX_{M,\widetilde{M}}
\Big[\sum_{j = 1}^{n} \sum_{i = 1}^{m}
\Pr\big[S_\gamma(M_{-j},\widetilde{M}_j) = i\big] X_j^{(i)}\Big]\\ &=
\gamma^2
\EX_{M,\widetilde{M}}
\Big[\sum_{j = 1}^{n} \sum_{i = 1}^{m}
\Pr\big[S_\gamma(M) = i\big] \widetilde{X}_j^{(i)}\Big]\\ &=
\gamma^2
\EX_{M}
\Big[\sum_{i = 1}^{m} \Pr\big[S_\gamma(M) = i\big] 
\EX_{\widetilde{M}}\big[\widetilde{X}^{(i)}\big]\Big]\\
&=
\gamma^2
\EX_{M}
\Big[\sum_{i = 1}^{m} \Pr\big[S_\gamma(M) = i\big] 
pn\Big] = \gamma^2 pn.
\end{align*}
We can conclude the lemma by plugging this bound into
(\ref{equ:accumatrix}).
\end{proof}

To obtain Lemma~\ref{lem:maxexp}, we set
$\gamma = 1+\frac{\sqrt{\ln m}}{\sqrt{n}}$. 
Now, Lemma~\ref{lem:maxgamma} gives
\begin{align*}
  \EX\big[\max\{X^{(1)},  \dots, X^{(m)}\}\big] &\leq 
  \left(1+\frac{\sqrt{\ln m}}{\sqrt{n}}\right)^2pn + 
  \frac{\ln m}{\ln\left(1+\frac{\sqrt{\ln m}}{\sqrt{n}}\right)}\\
&\leq 
  \left(1+\frac{3\sqrt{\ln m}}{\sqrt{n}}\right)pn + 
  \frac{\ln m}{\frac{\sqrt{\ln m}}{2 \sqrt{n}}},\\
\intertext{since $\frac{\sqrt{\ln m}}{\sqrt{n}} \leq 1$ by our assumption
$m \leq e^n$ and 
$\ln (1+x) \geq x/2$, for $x \in [0,1]$. Hence, using $pn \leq n$,}
  \EX\big[\max\{X^{(1)}, \dots, X^{(m)}\}\big] &\leq 
pn + 5\sqrt{n\ln m},
\end{align*}
as desired.

\section{Useful Consequences}\label{sec:conseq}

We now show several useful consequences of 
Theorem~\ref{thm:chernoff}. These results can be derived
directly from Theorem~\ref{thm:chernoff}, and therefore they
also hold for variants of the theorem with slightly different 
assumptions.

\subsection{The Lower Tail}

First, we show that an analogous bound holds for the
lower tail probability $\Pr[ X \leq (p-t) n]$.

\begin{corollary}\label{cor:chernoff_lower}
Let $X_1, \dots, X_n$ be independent random variables with 
$X_i \in \{0,1\}$ and $\Pr[X_i = 1] = p$, for $i = 1, \dots n$. 
Set $X \eqdef \sum_{i=1}^n X_i$. Then, for any $t \in [0, p]$, we have
\[
  \Pr[X \leq (p-t)n ] \leq e^{-\DKL(p-t \| p)n}.
\]
\end{corollary}

\begin{proof}
\begin{align*}
  \Pr[X \leq (p-t)n ] = \Pr[n-X \geq n-(p-t)n ] = \Pr[X' \geq (1-p+t)n ],
\end{align*}
where $X' = \sum_{i=1}^{n} X_i'$ with independent random variables
$X_i' \in \{0,1\}$ such that $\Pr[X_i' = 1]  = 1-p$. The 
result follows from $\DKL(1-p+t \| 1-p) = \DKL(p-t \| p)$.
\end{proof}

\subsection{Multiplicative Version}

Next, we derive a multiplicative variant of 
Theorem~\ref{thm:chernoff}. This well-known version of
the bound can be found in the classic text by Motwani and
Raghavan~\cite{MotwaniRa95}.

\begin{corollary}\label{cor:MR}
Let $X_1, \dots, X_n$ be independent random variables with 
$X_i \in \{0,1\}$ and $\Pr[X_i = 1] = p$, for $i = 1, \dots n$. Set 
$X \eqdef \sum_{i=1}^n X_i$ and
$\mu = pn$.  Then, for any $\delta \geq 0$, we have
\begin{align*}
  \Pr[X \geq (1+\delta)\mu ] &\leq 
    \left(\frac{e^{\delta}}{(1+\delta)^{1+\delta}}\right)^\mu,
  \text{ and}\\
  \Pr[X \leq (1-\delta)\mu ] &\leq 
   \left(\frac{e^{-\delta}}{(1-\delta)^{1-\delta}}\right)^\mu.
\end{align*}
\end{corollary}
\begin{proof}
Setting $t = \delta\mu/n$ in Theorem~\ref{thm:chernoff} yields
\begin{align*}
\Pr[X \geq (1+\delta)\mu] &\leq
\exp\left(-n\left[p(1+\delta)\ln(1+\delta) + 
p\left(\frac{1-p}{p}-\delta \right)\ln\left(1-\delta 
   \frac{p}{1-p}\right)\right]\right)\\
&= \left(\frac{(1-\delta p/(1-p))^{\delta - (1-p)/p}}{(1+\delta)^{1+\delta}}\right)^\mu\\
&\leq \left(\frac{e^{-\delta^2p/(1-p) + \delta}}{(1+\delta)^{1+\delta}}\right)^\mu
\leq \left(\frac{e^{\delta}}{(1+\delta)^{1+\delta}}\right)^\mu.
\end{align*}
Setting $t = \delta\mu/n$ in Corollary~\ref{cor:chernoff_lower} yields
\begin{align*}
\Pr[X \leq (1-\delta)\mu] &\leq
\exp\left(-n\left[p(1-\delta)\ln(1-\delta) + 
p\left(\frac{1-p}{p}+\delta \right)
\ln\left(1+\delta \frac{p}{1-p}\right)\right]\right)\\
&= \left(\frac{(1+\delta p/(1-p))^{-\delta - (1-p)/p}}
{(1-\delta)^{1-\delta}}\right)^\mu\\
&\leq \left(\frac{e^{-\delta^2p/(1-p) - \delta}}
{(1-\delta)^{1-\delta}}\right)^\mu
\leq \left(\frac{e^{-\delta}}{(1-\delta)^{1-\delta}}\right)^\mu.
\end{align*}
\end{proof}

\subsection{Useful Variants}

The next few corollaries give some handy variants of the bound
that are often more manageable in practice. First,
we give a simple bound for the multiplicative lower tail.

\begin{corollary}\label{cor:handy_lower}
  Let $X_1, \dots, X_n$ be independent random variables with 
  $X_i \in \{0,1\}$
  and $\Pr[X_i = 1] = p$, for $i = 1, \dots n$. Set 
  $X \eqdef \sum_{i=1}^n X_i$ and $\mu = pn$.
  Then, for any $\delta \in (0, 1)$, we have
  \[
    \Pr[X \leq (1-\delta)\mu ] \leq e^{-\delta^2\mu/2}.
  \]
\end{corollary}
\begin{proof}
By Corollary~\ref{cor:MR}
\[    
  \Pr[X \leq (1-\delta)\mu ] \leq 
  \left(\frac{e^{-\delta}}{(1-\delta)^{1-\delta}}\right)^\mu. 
\]
Using the power series expansion of $\ln(1-\delta)$, we get
\[
  (1-\delta)\ln(1-\delta) = -(1 - \delta) 
  \sum_{i=1}^\infty \frac{\delta^i}{i}
  = -\delta + \sum_{i=2}^\infty \frac{\delta^i}{(i-1)i}
  \geq -\delta + \delta^2/2.
\]
Thus,
\[    
  \Pr[X \leq (1-\delta)\mu ] \leq e^{[-\delta + \delta-\delta^2/2]\mu} = 
    e^{-\delta^2\mu/2},
\]
as claimed.
\end{proof}

An only slightly more complicated bound can be found for the
multiplicative upper tail.

\begin{corollary}\label{cor:handy_upper}
Let $X_1, \dots, X_n$ be independent random variables with 
$X_i \in \{0,1\}$ and $\Pr[X_i = 1] = p$, for $i = 1, \dots n$. 
Set $X \eqdef \sum_{i=1}^n X_i$ and $\mu = pn$.  Then, 
for any $\delta \geq 0$, we have
  \[
    \Pr[X \geq (1+\delta)\mu ] \leq e^{-\min\{\delta^2,\delta\}\mu/4}.
  \]
\end{corollary}
\begin{proof}
We may assume that $(1+\delta)p \leq 1$. Then, Theorem~\ref{thm:chernoff} 
gives
  \[
    \Pr[X \geq (1+\delta)pn ] \leq e^{-\DKL((1+\delta)p \| p)n}.
  \]
Define $f(\delta) \eqdef \DKL((1+\delta)p \| p)$.
Then,
\[
  f'(\delta) = p \ln(1+\delta) -  p \ln (1-\delta p / (1-p))
\]
and
\[
  f''(\delta) = \frac{p}{(1+\delta)(1-p - \delta p)} \geq 
  \frac{p}{1+\delta}. 
\]
By Taylor's theorem, we have
\[
  f(\delta) = f(0) + \delta f'(0) + \frac{\delta^2}{2} f''(\xi), 
\]
for some $\xi \in [0,\delta]$. Since $f(0) = f'(0) = 0$, it follows that
\[
  f(\delta) =  \frac{\delta^2}{2} f''(\xi) \geq 
  \frac{\delta^2p}{2(1+\xi)} \geq \frac{\delta^2p}{2(1+\delta)}.
\]
For $\delta \geq 1$, we have $\delta/(1+\delta) \geq 1/2$, for
$\delta < 1$, we have $1/(\delta+1) \geq 1/2$. This gives, for 
all $\delta \geq 0$,
\[
  f(\delta) \geq  \min\{\delta^2, \delta\}p/4,
\]
and the claim follows.
\end{proof}

The following corollary combines the two bounds. This variant 
can be found, e.g., in the book by Arora and Barak~\cite{AroraBa09}.

\begin{corollary}
Let $X_1, \dots, X_n$ be independent random variables with 
$X_i \in \{0,1\}$ and $\Pr[X_i = 1] = p$, for $i = 1, \dots n$. 
Set $X \eqdef \sum_{i=1}^n X_i$ and $\mu = pn$.  Then, for any 
$\delta > 0$, we have
\[
 \Pr[|X - \mu| \geq \delta\mu ] \leq 2e^{-\min\{\delta^2, \delta\}\mu/4}.
\]
\end{corollary}

\begin{proof}
Combine Corollaries~\ref{cor:handy_lower} and~\ref{cor:handy_upper}.
\end{proof}

The following corollary, which appears, e.g., in the book by
Motwani and Raghavan~\cite{MotwaniRa95}, is also sometimes useful.
\begin{corollary}
  Let $X_1, \dots, X_n$ be independent random variables with 
  $X_i \in \{0,1\}$ and $\Pr[X_i = 1] = p$, for $i = 1, \dots n$. 
  Set $X \eqdef \sum_{i=1}^n X_i$ and $\mu = pn$. For $t \geq 2e\mu$, we 
  have
  \[
    \Pr[X \geq t ] \leq 2^{-t}.
  \]
\end{corollary}

\begin{proof}
By Corollary~\ref{cor:MR}
\[    
  \Pr[X \geq (1+\delta)\mu ] \leq \left(\frac{e^{\delta}}{(1+\delta)^{1+\delta}}\right)^\mu 
    \leq \left(\frac{e}{1+\delta}\right)^{(1+\delta)\mu}.
\]
For $\delta \geq 2e-1$, the denominator in the right hand side is at least $2e$, and the
claim follows.
\end{proof}

\section{Generalizations}

We mention a few generalizations of the proof techniques
for Section~\ref{sec:proofs}. Since the consequences from
Section~\ref{sec:conseq} are based on simple algebraic manipulation
of the bounds, the same consequences also hold for the generalized
settings.

\subsection{Hoeffding Extension}

The moment method (Section~\ref{sec:moment}) yields many generalizations
of Theorem~\ref{thm:chernoff}. The following result is known as
\emph{Hoeffding's extension}~\cite{Hoeffding63}. It shows that
the $X_i$ can actually be chosen to be continuous with varying 
expectations.
\begin{theorem}\label{thm:chernoff_hoeff}
  Let $X_1, \dots, X_n$ be independent random variables with 
  $X_i \in [0,1]$ and $\EX[X_i] = p_i$.
  Set $X \eqdef \sum_{i=1}^n X_i$ and $p \eqdef (1/n)\sum_{i=1}^n p_i$.
  Then, for any $t \in [0, 1-p]$, we have
  \[
    \Pr[X \geq (p+t)n ] \leq e^{-\DKL(p+t \| p)n}.
  \]
\end{theorem}

\begin{proof}
  Let $\lambda > 0$ a parameter to be determined later. As before,
  Markov's inequality yields
  \[
  \Pr\bigl[e^{\lambda X} \geq e^{\lambda (p+t)n} \bigr] \leq
    \frac{\EX[e^{\lambda X}]}{e^{\lambda (p+t)n}}.
  \]
  Using independence, we get
  \begin{equation}\label{equ:boundExp}
    \EX[e^{\lambda X}] = \EX\Bigl[e^{\lambda \sum_{i=1}^n X_i}\Bigr]
    =  \prod_{i=1}^n \EX\Bigl[e^{\lambda X_i}\Bigr].
  \end{equation}
  Now we need to estimate $\EX\bigl[e^{\lambda X_i}\bigr]$.
  The function $z \mapsto e^{\lambda z}$ is convex,
  so $e^{\lambda z} 
  \leq (1-z) e^{0 \cdot \lambda} + z e^{1 \cdot \lambda}$ for 
  $z \in [0,1]$. Hence,
  \[
    \EX\bigl[e^{\lambda X_i}\bigr] \leq \EX[1-X_i + X_i e^\lambda] = 
      1-p_i + p_ie^\lambda.
  \]
  Going back to (\ref{equ:boundExp}),
  \[
    \EX[e^{\lambda X}] \leq 
      \prod_{i=1}^n (1-p_i + p_i e^\lambda).
  \]
  Using the arithmetic-geometric mean inequality $\prod_{i=1}^n x_i \leq
    \bigl((1/n)\sum_{i=1}^n x_i\bigr)^n$, for $x_i \geq 0$, this is
  \[
    \EX[e^{\lambda X}] \leq (1-p+pe^\lambda)^n.
  \]
  From here we continue as in Section~\ref{sec:moment}.
\end{proof}

\subsection{Hypergeometric Distribution}

Chv\'atals proof~\cite{Chvatal79} from Section~\ref{sec:chvatal} 
generalizes to the \emph{hypergeometric} distribution.
We emphasize once again that this means that all the corollaries
from Section~\ref{sec:conseq} also apply to this case.
\begin{theorem}\label{thm:hyper}
  Suppose we have an urn with $N$ balls, $P$ of which are red.
  We randomly draw $n$ balls from the urn without replacement.
  Let $H(N,P,n)$ denote the number of red balls in the sample.
  Set $p \eqdef P/N$.
  Then, for any $t \in [0, 1-p]$, we have
  \[
    \Pr\big[H(N,P,n) \geq (p+t)n \big] \leq e^{-\DKL(p+t \| p)n}.
  \]
\end{theorem}

\begin{proof}
It is well known that
\[
\Pr[H(N,P,n) = l] = \binom{P}{l}\binom{N-p}{n-l}\binom{N}{l}^{-1},
\]
for $l = 0, \dots, n$. 
\begin{claim}\label{clm:hyper_bound}
For every $j \in \{0, \dots, n\}$, we have
\[
\binom{N}{n}^{-1} \sum_{i = j}^n \binom{P}{i}\binom{N-P}{n - i}\binom{i}{j}
\leq \binom{n}{j}p^j.
\]
\end{claim}
\begin{proof}
Consider the following random experiment: take a random permutation of the $N$
balls in the urn. Let $S$ be the sequence of the first $n$ elements in the permutation.
Let $X$ be the number of $j$-subsets of $S$ that contain only red balls.
We compute $\EX[X]$ in two different ways. On the one hand, 
\begin{equation}\label{equ:hyper_bound_var1}
\EX[X] = \sum_{i = j}^n \Pr[\text{S contains $i$ red balls}] \binom{i}{j}
  = \sum_{i = j}^n \binom{N}{n}^{-1} \binom{P}{i} \binom{N-P}{n-i} \binom{i}{j}.
\end{equation}
On the other hand, let $I \subseteq \{1, \dots, n\}$ with 
$|I| = j$. Then the probability that all the balls in the positions
indexed by $I$ are red is
\[
\frac{P}{N} \cdot \frac{P-1}{N-1} \cdot \cdots \cdot
\frac{P-j+1}{N-j+1} \leq \left(\frac{P}{N}\right)^j = p^j. 
\]
Thus, by linearity of expectation $\EX[X] \leq \binom{n}{j} p^j$. Together with
(\ref{equ:hyper_bound_var1}), the claim follows.
\end{proof}
\begin{claim}\label{clm:hyper_bound2}
For every $\tau \geq 1$, we have
\[
\binom{N}{n}^{-1} \sum_{i = 0}^n \binom{P}{i}\binom{N-P}{n - i}\tau^i
\leq (1 + (\tau-1)p)^n.
\]
\end{claim}
\begin{proof}
Using Claim~\ref{clm:hyper_bound} and the Binomial theorem (twice),
\begin{align*}
\binom{N}{n}^{-1} \sum_{i = 0}^n \binom{P}{i}\binom{N-P}{n - i}\tau^i
&= \binom{N}{n}^{-1} \sum_{i = 0}^n \binom{P}{i}\binom{N-P}{n - i}
  (1-(\tau-1))^i\\
&= \binom{N}{n}^{-1} \sum_{i = 0}^n \binom{P}{i}\binom{N-P}{n - i}
  \sum_{j=0}^i \binom{i}{j}(\tau-1)^j\\ 
&= \binom{N}{n}^{-1} \sum_{j = 0}^n (\tau-1)^j \sum_{i=j}^n 
  \binom{P}{i}\binom{N-P}{n - i}\binom{i}{j}\\
&\leq \sum_{j = 0}^n \binom{n}{j}((\tau-1)p)^j = (1 + (\tau-1)p)^n,
\end{align*}
as claimed.
\end{proof}

Thus, for any $\tau \geq 1$ and $k \geq pn$, we get as before
\begin{multline*}
\Pr[H(N,P,n)\geq k] = \binom{N}{n}^{-1}\sum_{i=k}^n \binom{P}{i}\binom{N-P}{n-i}\\
\leq \binom{N}{n}^{-1}\sum_{i=0}^n \binom{P}{i}\binom{N-P}{n-i}
\tau^{i-k}\leq
\frac{(p\tau+1-p)^n}{\tau^k},
\end{multline*}
by Claim~\ref{clm:hyper_bound2}. From here the proof proceeds as in 
Section~\ref{sec:chvatal}.
\end{proof}

\subsection{Negative Correlations}

The proof by Impagliazzo and Kabanets~\cite{ImpagliazzoKa10} from
Section~\ref{sec:ik} can be used to relax the independence assumption.
It now suffices that the random variables are \emph{negatively
correlated}.
\begin{theorem}\label{thm:general_ik}
  Let $X_1, \dots, X_n$ be random variables with $X_i \in \{0,1\}$.
  Suppose there exist $p_i \in [0,1]$, $i = 1, \dots, n$, such 
  that for every index set $I \subseteq \{1, \dots, n\}$, we
  have $\EX\big[\prod_{i \in I} X_i \big] \leq \prod_{i \in I} p_i$.
  Set $X \eqdef \sum_{i=1}^n X_i$ and $p \eqdef (1/n)\sum_{i=1}^n p_i$.
  Then, for any $t \in [0, 1-p]$, we have
  \[
    \Pr[X \geq (p+t)n ] \leq e^{-\DKL(p+t \| p)n}.
  \]
\end{theorem}

\begin{proof}
Let $\lambda \in [0,1]$ be a parameter to be chosen later. 
Let $I \subseteq \{1, \dots, n\}$ be a random index set obtained by
including each element $i \in \{1, \dots, n\}$ with probability 
$\lambda$.
As before, we estimate the expectation
$\EX\bigl[\prod_{i \in I} X_i \bigr]$ in two different ways, where
the expectation is over the random choice of $X_1, \dots, X_n$ and $I$.
Similarly to before,
\begin{multline}\label{equ:ikupper_gen}
\EX\Bigl[\prod_{i \in I} X_i \Bigr]
= \sum_{S \subseteq \{1, \dots, n\}} 
     \Pr[I = S] \cdot \EX \Bigl[ \prod_{i \in S} X_i \Bigr]
\leq \sum_{S \subseteq \{1, \dots, n\}} \lambda^{|S|}(1-\lambda)^{n-|S|} 
\cdot \Big(\prod_{i \in S} p_i \Big)\\
= \sum_{S \subseteq \{1, \dots, n\}} 
\Big(\prod_{i \in S} \lambda p_i\Big)
\Big(\prod_{i \in \{1, \dots, n\} \setminus S} 
(1-\lambda)\Big)  =
\prod_{i=1}^n (1-\lambda + p_i \lambda)
\leq (1-\lambda + p \lambda)^n,
\end{multline}
by the 
arithmetic-geometric mean inequality.
The proof of the lower bound remains unchanged and yields
\[
\EX\Bigl[\prod_{i \in I} X_i\Bigr]
 \geq
(1-\lambda)^{(1-p-t)n}\Pr[X \geq (p+t)n],
\]
as before. Combining with (\ref{equ:ikupper_gen}) and optimizing
for $\lambda$ finishes the proof, see Section~\ref{sec:ik}.
\end{proof}

\paragraph*{Acknowledgments.}
This survey is based on lecture notes for a class
on advanced algorithms at Freie Universit\"at Berlin.
I would like to thank all the students who took
this class for their interest and participation. I would also
like to thank Nabil Mustafa and Jonathan Ullman for 
valuable comments that improved this survey.

\bibliographystyle{abbrv}
\bibliography{chernoff}
\end{document}